\documentclass[publicdomain]{eptcs}
\providecommand{\event}{MCU 2013}

\usepackage{amsfonts}
\usepackage{latexsym}
\usepackage{amsmath}
\usepackage{doi}

\input{tcilatex}

\begin{document}

\title{How to Obtain Computational Completeness in \\
P Systems with One Catalyst}

\author{Rudolf Freund 
\institute{
Technische Universit\"{a}t Wien, Institut f\"{u}r Computersprachen\\
Favoritenstr. 9, A-1040 Wien, Austria}
\email{rudi@emcc.at}
\and Gheorghe P\u{a}un 
\institute{Institute of Mathematics of the Romanian Academy\\
PO Box 1-764, 014700 Bucure\c sti, Romania, and\\[2mm]
Department of Computer Science and Artificial Intelligence\\
University of Sevilla\\
Avda. Reina Mercedes s/n, 41012 Sevilla, Spain}
\email{gpaun@us.es, ghpaun@gmail.com}
}

\def\titlerunning{How to Obtain Computational Completeness in 
P Systems with One Catalyst} 
\def\authorrunning{Rudolf Freund and Gheorghe P\u{a}un}
\maketitle

\begin{abstract}
Whether P systems with only one catalyst can already be computationally
complete, is still an open problem. Here we establish computational
completeness by using specific variants of additional control mechanisms. At
each step using only multiset rewriting rules from one set of a finite
number of sets of multiset rewriting rules allows for obtaining
computational completeness with one catalyst and only one membrane. If the
targets are used for choosing the multiset of rules to be applied, for
getting computational completeness with only one catalyst more than one
membrane is needed. If the available sets of rules change periodically with
time, computational completeness can be obtained with one catalyst in one
membrane. Moreover, we also improve existing computational completeness
results for P systems with mobile catalysts and for P systems with membrane
creation.
\end{abstract}

\section{Introduction}

P systems with catalytic rules were already considered in the originating
papers for membrane systems, see \cite{cell1}. In \cite{Freundetal2005} two
catalysts were shown to be sufficient for getting computational completeness
(throughout this paper, with this notion we will indicate that all
recursively enumerable sets of (vectors of) non-negative integers can be
generated). Since then, it has become one of the most challenging open
problems in the area of P systems, whether or not one catalyst might already
be enough to obtain computational completeness.

Using additional control mechanisms as, for example, priorities or
promoters/inhibitors, P systems with only one catalyst can be shown to be
computationally complete, e.g., see Chapter 4 of \cite{handMC}. On the other
hand, additional features for the catalyst may be taken into account; for
example, we may use bi-stable catalysts (catalysts switching between two
different states) or mobile catalysts (catalysts able to cross membranes).
Moreover, additional membrane features may be used, for example, membrane
creation or controlling the membrane permeability by means of the operations 
$\delta $ and $\tau $.

P systems with mobile catalysts were introduced in \cite{KrishnaAndrei}, and
their computational completeness was proved with using three membranes and
targets of the forms $here$, $out$, and $in_{j}$. We here improve this
result by replacing the targets $in_{j}$ with the weaker one $in$.

P systems with membrane creation were introduced in \cite{MadhuKamala},
showing both their computational completeness and efficiency (the
Hamiltonian path problem is solved in linear time in a semi-uniform way;
this result was improved in \cite{MiguelMario}, where a polynomial solution
to the Subset Sum problem in a uniform way is provided). For proving
computational completeness, in \cite{MadhuKamala} (Theorem 2) P systems
starting with one membrane, having four membranes at some time during the
computation, using one catalyst, and also controlling the membrane
permeability by means of the operations $\delta $ and $\tau $ are needed.
However, as already shown in \cite{deltatau}, P systems with one catalyst
and using the operations $\delta $ and $\tau $ are computationally complete,
hence, the membrane creation facility is not necessary for getting
computational completeness in this framework. Here we improve the result
shown in \cite{MadhuKamala} from two points of view: (i) the control of
membrane permeability is not used, and (ii) the maximal number of membranes
used during a computation is two.

Recenty, several variants of P systems using only one catalyst together with
control mechanisms for choosing the rules applicable in a computation step
have been considered: for example, in \cite{control} the rules are labeled
with elements from an alphabet $H$ and in each step a maximal multiset of
rules having the same label from $H$ is applied. In this paper, we will give
a short proof for the computational completeness of these \textit{P systems
with label selection }with only one catalyst in a single membrane. As a
specific variant, for each membrane we can choose the rules according to the
targets, and we will prove computational completeness for these \textit{P
systems with target selection }with only one catalyst, but needing more than
one membrane (such systems with only one membrane lead to the still open
problem of catalytic P systems with one catalyst).

Regular control languages were considered already in \cite{control} for the
maximally parallel derivation mode, whereas in \cite{TimeVaryingP}
computational completeness was proved for the sequential mode: there even
only non-cooperative rules were needed in one membrane for time-varying P
systems to obtain computational completeness (in time-varying systems, the
set of available rules varies periodically with time, i.e., the regular
control language is of the very specific form $W=\left( U_{1}\dots
U_{p}\right) ^{\ast }$, allowing to apply rules from a set $U_{i}$ in the
computation step $pn+i$, $n\geq 0$; $p$ is called the \textit{period}), but
a bounded number of steps without applying any rule had to be allowed. We
here prove that \textit{time-varying P systems} using the maximally parallel
derivation mode in one membrane with only one catalyst are computationally
complete with a period of six and the usual halting when no rule can be
applied. 

The new results exhibited in this paper first were presented in \cite%
{FreundPaun2013}. For the newest developments in the area of P 
systems we refer the reader to the P systems website \cite{webP}.

\section{Prerequisites}

The set of non-negative integers is denoted by 
$\mathbb{N}$. An \textit{alphabet }$V$
is a finite non-empty set of abstract \textit{symbols}. Given $V$, the free
monoid generated by $V$ under the operation of concatenation is denoted by $%
V^{\ast }$; the elements of $V^{\ast }$ are called strings, and the \textit{%
empty string} is denoted by $\lambda $; $V^{\ast }\setminus \left\{ \lambda
\right\} $ is denoted by $V^{+}$. Let $\left\{ a_{1},\cdots ,a_{n}\right\} $
be an arbitrary alphabet; the number of occurrences of a symbol $a_{i}$ in a
string $x$ is denoted by $\left\vert x\right\vert _{a_{i}}$. For a fixed
sequence $\left\langle a_{1},\cdots ,a_{n}\right\rangle $ of the symbols in
the alphabet $\left\{ a_{1},\cdots ,a_{n}\right\} $,  the \textit{Parikh
vector} associated with $x$ with respect to $\left\langle a_{1},\cdots
,a_{n}\right\rangle $ is $\left( \left\vert x\right\vert _{a_{1}},\cdots
,\left\vert x\right\vert _{a_{n}}\right) $; the \textit{Parikh image} of a
language $L$ over $\left\{ a_{1},\cdots ,a_{n}\right\} $ is the set of all
Parikh vectors of strings in $L$, and we denote it by $Ps\left( L\right) $.
For a family of languages $FL$, the family of Parikh images of languages in $%
FL$ is denoted by $PsFL$; for families of languages of a one-letter
alphabet, the corresponding sets of non-negative integers are denoted by $NFL
$.

A (finite) multiset over the (finite) alphabet $V$, $V=\left\{ a_{1},\cdots
,a_{n}\right\} $, is a mapping $f:V\longrightarrow \mathbb{N}$ and
represented by $\left\langle f\left( a_{1}\right) ,a_{1}\right\rangle \cdots
\left\langle f\left( a_{n}\right) ,a_{n}\right\rangle $ or by any string $x$
the Parikh vector of which with respect to $\left\langle a_{1},\cdots
,a_{n}\right\rangle $ is 
$\left( f\left( a_{1}\right) ,\cdots ,f\left( a_{n}\right) \right) $. In the
following we will not distinguish between a vector $\left( m_{1},\cdots
,m_{n}\right) ,$ its representation by a multiset $\left\langle
m_{1},a_{1}\right\rangle \cdots \left\langle m_{n},a_{n}\right\rangle $ or
its representation by a string $x$ having the Parikh vector $\left(
\left\vert x\right\vert _{a_{1}},\cdots ,\left\vert x\right\vert
_{a_{n}}\right) =\left( m_{1},\cdots ,m_{n}\right) $. For a fixed sequence $%
\left\langle a_{1},\cdots ,a_{n}\right\rangle $ of the symbols in the
alphabet $\left\{ a_{1},\cdots ,a_{n}\right\} $, the representation of the
multiset $\left\langle m_{1},a_{1}\right\rangle \cdots \left\langle
m_{n},a_{n}\right\rangle $ by the string $a_{1}^{m_{1}}\cdots a_{n}^{m_{n}}$
is unique. 

The family of regular and recursively enumerable string languages is denoted
by $REG$ and $RE$, respectively. For more details of formal language theory
the reader is referred to the monographs and handbooks in this area as \cite%
{DassowPaun1989} and \cite{HandbookFormalLanguages1997}.

\smallskip

A \emph{register machine} is a tuple $M=\left( m,B,l_{0},l_{h},P\right) $,
where $m$ is the number of registers, $P$ is the set of instructions
bijectively labeled by elements of $B$, $l_{0}\in B$ is the initial label,
and $l_{h}\in B$ is the final label. The instructions of $M$ can be of the
following forms:

\begin{itemize}
\item $l_{1}:\left( \mathtt{ADD}\left( j\right) ,l_{2},l_{3}\right) $, with $l_{1}\in
B\setminus \left\{ l_{h}\right\} $, $l_{2},l_{3}\in B$, $1\leq j\leq m$.
\newline
Increase the value of register $j$ by one, and non-deterministically jump to
instruction $l_{2}$ or $l_{3}$. This instruction is usually called \emph{%
increment}.

\item $l_{1}:\left( \mathtt{SUB}\left( j\right) ,l_{2},l_{3}\right) $, with $l_{1}\in
B\setminus \left\{ l_{h}\right\} $, $l_{2},l_{3}\in B$, $1\leq j\leq m$.
\newline
If the value of register $j$ is zero then jump to instruction $l_{3}$,
otherwise decrease the value of register $j$ by one and jump to instruction $%
l_{2}$. The two cases of this instruction are usually called \emph{zero-test}
and \emph{decrement}, respectively.

\item $l_{h}: \mathtt{HALT}$. Stop the execution of instructions of the register
machine.
\end{itemize}

A \emph{configuration} of a register machine is described by the contents of
each register and by the value of the current label, which indicates the
next instruction to be executed. Computations start by executing the first
instruction of $P$ (labeled with $l_{0}$), and terminate with reaching the 
$\mathtt{HALT}$ instruction.

Register machines provide a simple computing model which
is computationally complete (e.g., see \cite{Minsky1967}). 
For generating sets of vectors of non-negative
integers, we start with empty registers, use the first two registers for the
necessary computations and take as results the contents of the $k$ registers 
$3$ to $k+2$ in all possible halting computations; during a computation of $M
$, only the registers $1$ and $2$ can be decremented, and moreover, we
assume the registers $1$ and $2$ to be empty at the end of a halting
computation. In the following, we shall call a specific model of P systems 
\emph{computationally complete} if and only if for any register machine $M$
we can effectively construct an equivalent P system $\Pi $ of that type
simulating each step of $M$ in a bounded number of steps and yielding the
same results.

\subsection{P Systems}

\label{Psystems}

The basic ingredients of a (cell-like) P system are the membrane structure,
the objects placed in the membrane regions, and the evolution rules. The 
\textit{membrane structure} is a hierarchical arrangement of membranes. Each
membrane defines a \textit{region/compartment}, the space between the
membrane and the immediately inner membranes; the outermost membrane is
called the \textit{skin membrane}, the region outside is the \textit{%
environment}, also indicated by (the label) $0$. Each membrane can be
labeled, and the label (from a set $Lab$) will identify both the membrane
and its region. The membrane structure can be represented by a rooted tree
(with the label of a membrane in each node and the skin in the root), but
also by an expression of correctly nested labeled parentheses. The \textit{%
objects} (multisets) are placed in the compartments of the membrane
structure and usually represented by strings, with the multiplicity of a
symbol corresponding to the number of occurrences of that symbol in the
string. \textit{The evolution rules} are multiset rewriting rules of the
form $u\rightarrow v$, where $u$ is a multiset of objects from a given set $O
$ and $v=\left( b_{1},tar_{1}\right) \dots \left( b_{k},tar_{k}\right) $
with $b_{i}\in O$ and $tar_{i}\in \left\{ here,out,in\right\} $ or $%
tar_{i}\in \left\{ here,out\right\} \cup \left\{ in_{j}\mid j\in Lab\right\} 
$, $1\leq i\leq k$. Using such a rule means \textquotedblleft consuming" the
objects of $u$ and \textquotedblleft producing" the objects $b_{1},\dots
,b_{k}$ of $v$; the \textit{target indications} (\textit{targets} for short) 
$here$, $out$, and $in$ mean that an object with the target $here$ remains
in the same region where the rule is applied, an object with the target $out$
is sent out of the respective membrane (in this way, objects can also be
sent to the environment, when the rule is applied in the skin region), while
an object with the target $in$ is sent to one of the immediately inner
membranes, non-deterministically chosen, wheras with $in_{j}$ this inner
membrane can be specified directly. Usually, we omit the target $here$. With
respect to the tree representation of the membrane structure of the P
system, the target $out$ means moving the object to the region represented
by the parent node, and the target $in$ means moving the object to a region
represented by one of the children nodes; with the target $in_{j}$ we can
directly specify which of the children nodes is to be chosen.

Formally, a (cell-like) P system is a construct 
\begin{equation*}
\Pi =(O,\mu ,w_{1},\dots ,w_{m},R_{1},\dots ,R_{m},f)
\end{equation*}%
where $O$ is the alphabet of objects, $\mu $ is the membrane structure (with 
$m$ membranes), $w_{1},\dots ,w_{m}$ are multisets of objects present in the 
$m$ regions of $\mu $ at the beginning of a computation , $R_{1},\dots
,R_{m} $ are finite sets of evolution rules, associated with the regions of $%
\mu $, and $f$ is the label of the membrane region from which the outputs
are taken ($f=0$ indicates that the output is taken from the environment).

If a rule $u\rightarrow v$ has at least two objects in $u$, then it is
called \textit{cooperative}, otherwise it is called \textit{non-cooperative}. 
In \textit{catalytic P systems} we use non-cooperative as well as \textit{%
catalytic rules} which are of the form $ca\rightarrow cv$, where $c$ is a
special object -- a so-called \textit{catalyst} -- which never evolves and
never passes through a membrane (both these restrictions can be relaxed),
but it just assists object $a$ to evolve to the multiset $v$. In a \textit{%
purely catalytic P system} we only allow catalytic rules. In both catalytic
and purely catalytic P systems we replace $O$ by $O,C$ in order to specify
those objects from $O$ which are the catalysts in the set $C$, i.e., we
write 
\begin{equation*}
\Pi =(O,C,\mu ,w_{1},\dots ,w_{m},R_{1},\dots ,R_{m},f).
\end{equation*}

The evolution rules are used in the \textit{non-deterministic maximally
parallel} way, i.e., in any computation step of $\Pi $ we choose a multiset
of rules from the sets $R_{1},\dots ,R_{m}$ in such a way that no further
rule can be added to it so that the obtained multiset would still be
applicable to the existing objects in the membrane regions $1,\dots ,m$.

The membranes and the objects present in the compartments of a system at a
given time form a \textit{configuration}; starting from a given \textit{%
initial configuration} and using the rules as explained above, we get 
\textit{transitions} among configurations; a sequence of transitions forms a 
\textit{computation}. A computation is \textit{halting} if it reaches a
configuration where no rule can be applied anymore. With a halting
computation we associate a \textit{result}, in the form of the number of
objects present in membrane $f$ in the halting configuration. The set of
non-negative integers and the set of (Parikh) vectors of non-negative
integers obtained as results of halting computations in $\Pi $ are denoted
by $N\left( \Pi \right) $ and $Ps\left( \Pi \right) $, respectively.

\smallskip

The family of sets $Y\left( \Pi \right) $, $Y\in \left\{ N,Ps\right\} $,
computed by P systems with at most $m$ membranes and cooperative rules and
with non-cooperative rules is denoted by $YOP_{m}\left( coop\right) $ and $%
YOP_{m}\left( ncoo\right) $, respectively. It is well known that for any $%
m\geq 1$, $YREG=YOP_{m}\left( ncoo\right) \subset NOP_{m}\left( coop\right)
=YRE$, see \cite{cell1}.

\smallskip

The family of sets $Y\left( \Pi \right) $, $Y\in \left\{ N,Ps\right\} $,
computed by (purely) catalytic P systems with at most $m$ membranes and at
most $k$ catalysts is denoted by $YOP_{m}\left( cat_{k}\right) $ ($%
YOP_{m}\left( pcat_{k}\right) $); from \cite{Freundetal2005} we know that,
with the results being sent to the environment, we have $YOP_{1}\left(
cat_{2}\right) =YOP_{1}\left( pcat_{3}\right) =YRE$.

\smallskip

If we allow catalysts to move from one membrane region to another one, then
we speak of \textit{P systems with mobile catalysts}. The families of sets $%
N\left( \Pi \right) $ and $Ps\left( \Pi \right) $ computed by P systems with
at most $m$ membranes and $k$ mobile catalysts are denoted by $NOP_{m}\left(
mcat_{k}\right) $ and $PsOP_{m}\left( mcat_{k}\right) $, respectively.

\smallskip

For all the variants of P systems using rules of some type $X$ as defined
above, we may consider systems containing only rules of the form $%
u\rightarrow v$ where $u\in O$ and $v=\left( b_{1},tar\right) \dots \left(
b_{k},tar\right) $ with $b_{i}\in O$ and $tar\in \left\{ here,out,in\right\} 
$ or $tar\in \left\{ here,out\right\} \cup \left\{ in_{j}\mid j\in H\right\} 
$, $1\leq i\leq k$, i.e., in each rule there is only one target for all
objects $b_{i}$; moreover, with the target $in$ we assume all objects
generated by the rules of the chosen multiset of rules applied to the
objects in a specific region of the current configuration to choose the same
inner membrane. If \textit{catalytic rules} are considered, then we request
the rules to be of the form $ca\rightarrow \left( c,here\right) \left(
b_{1},here\right) \dots \left( b_{k},here\right) $, as the catalyst is not
allowed to move. \textit{P systems with target selection} contain only these
forms of rules; moreover, in each computation step, for each membrane region 
$i$ we choose a non-empty multiset (if it exists) of rules $R_{i}^{\prime }$
from $R_{i}$ having the same target $tar$ -- for different membranes these
targets may be different -- and apply $R_{i}^{\prime }$ in the maximally
parallel way, i.e., the set $R_{i}^{\prime }$ cannot be extended by any
further rule from $R_{i}$ with the target $tar$ so that the obtained
multiset of rules would still be applicable to the existing objects in the
membrane region $i$. The family of sets $N\left( \Pi \right) $ and $Ps\left(
\Pi \right) $ computed by P systems with target selection with at most $m$
membranes and rules of type $X$ is denoted by $NOP_{m}\left( X,ts\right) $
and $PsOP_{m}\left( X,ts\right) $, respectively.

\smallskip

For all the variants of P systems of type $X$, we may consider to label all
the rules in the sets $R_{1},\dots ,R_{m}$ in a one-to-one manner by labels
from a set $H$ and to take a set $W$ containing subsets of $H$. Then a 
\textit{P system with label selection }is a construct 
\begin{equation*}
\Pi =(O,\mu ,w_{1},\dots ,w_{m},R_{1},\dots ,R_{m},H,W,f)
\end{equation*}%
where $\Pi ^{\prime }=(O,\mu ,w_{1},\dots ,w_{m},R_{1},\dots ,R_{m},f)$ is a
P system as defined above, $H$ is a set of labels for the rules in the sets $%
R_{1},\dots ,R_{m}$, and $W\subseteq 2^{H}$. In any transition step in $\Pi $
we first select a set of labels $U\in W$ and then apply a non-empty multiset 
$R$ of rules such that all the labels of these rules in $R$ are in $U$ in
the maximally parallel way, i.e., the set $R$ cannot be extended by any
further rule with a label from $U$ so that the obtained multiset of rules
would still be applicable to the existing objects in the membrane regions $%
1,\dots ,m$. The family of sets $N\left( \Pi \right) $ and $Ps\left( \Pi
\right) $ computed by P systems with label selection with at most $m$
membranes and rules of type $X$ is denoted by $NOP_{m}\left( X,ls\right) $
and $PsOP_{m}\left( X,ls\right) $, respectively.

Another method to control the application of the labeled rules is to use
control languages (see \cite{control} and \cite{TimeVaryingP}). A \textit{%
controlled P system} is a construct 
\begin{equation*}
\Pi =(O,\mu ,w_{1},\dots ,w_{m},R_{1},\dots ,R_{m},H,L,f)
\end{equation*}%
where $\Pi ^{\prime }=(O,\mu ,w_{1},\dots ,w_{m},R_{1},\dots ,R_{m},f)$ is a
P system as defined above, $H$ is a set of labels for the rules in the sets $%
R_{1},\dots ,R_{m}$, and $L$ is a string language over $2^{H}$ (each subset
of $H$ represents an element of the alphabet for $L$) from a family $FL$.
Every successful computation in $\Pi $ has to follow a control word $%
U_{1}\dots U_{n}\in L$: in transition step $i$, only rules with labels in $%
U_{i}$ are allowed to be applied (but again in the maximally parallel way,
i.e., we have to apply a multiset $R$ of rules with labels in $U_{i}$ which
cannot be extended by any rule with a label in $U_{i}$ such that the
resulting multiset would still be applicable), and after the $n$-th
transition, the computation halts; we may relax this end condition, i.e., we
may stop after the $i$-th transition for any $i\leq n$, and then we speak of 
\textit{weakly controlled P systems}. If $L=\left( U_{1}\dots U_{p}\right)
^{\ast }$, $\Pi $ is called a \textit{(weakly) time-varying P system}: in
the computation step $pn+i$, $n\geq 0$, rules from the set $U_{i}$ have to
be applied; $p$ is called the \textit{period}. The family of sets $Y\left(
\Pi \right) $, $Y\in \left\{ N,Ps\right\} $, computed by (weakly) controlled
P systems and (weakly) time-varying P systems with period $p$, with at most $%
m$ membranes and rules of type $X$ as well as control languages in $FL$ is
denoted by $YOP_{m}\left( X,C\left( FL\right) \right) $ ($YOP_{m}\left(
X,wC\left( FL\right) \right) $) and $YOP_{m}\left( X,TV_{p}\right) $ ($%
YOP_{m}\left( X,wTV_{p}\right) $), respectively.

\medskip

In the \textit{P systems with membrane creation} considered in this paper,
besides the catalytic rules $ca\rightarrow c\left( u,tar\right) $ and the
non-cooperative rules $a\rightarrow \left( u,tar\right) $ we also use
catalytic membrane creation rules of the form $ca\rightarrow c[\ u{{\ ]_{{}}}%
_{{}}}_{i}$ (in the context of $c$, from the object $a$ a new membrane with
label $i$ containing the multiset $u$ is generated) and membrane dissolution
rules $a\rightarrow u\delta $ (we assume that no objects can be sent into a
membrane which is going to be dissolved; with dissolving the membrane $i$ by
applying $\delta $, all objects contained inside this membrane are collected
in the region surrounding the dissolved membrane); in all cases, $c$ is a
catalyst, $a$ is an object, $u$ is a multiset, and $tar$ is a target of the
form $here$, $out$, and $in_{j}$. The family of sets $Y\left( \Pi \right) $, 
$Y\in \left\{ N,Ps\right\} $, computed by such P systems with membrane
creation and using at most $k$ catalysts, with $m$ initial membranes and
having at most $h$ membranes during its computations is denoted by $%
YP_{m,h}\left( cat_{k},mcre\right) $.

\section{Computational Completeness of P Systems with Label Selection}

\begin{theorem}
\label{TheoremLS} $YOP_{1}\left( cat_{1},ls\right) =YRE$, $Y\in \left\{
N,Ps\right\} $.
\end{theorem}

\begin{proof}
We only prove the inclusion $PsRE\subseteq PsOP_{1}\left( cat_{1},ls\right) $%
. Let us consider a register machine $M=\left( n+2,B,l_{0},l_{h},I\right) $
with only the first and the second register ever being decremented, and let $%
A=\left\{ a_{1},\dots ,a_{n+2}\right\} $ be the set of objects for
representing the contents of the registers $1$ to $n+2$ of $M$. We construct
the following P system:%
\begin{eqnarray*}
\Pi  &=&(O,\left\{ c\right\} ,[\ {{\ ]_{{}}}_{{}}}_{1},cdl_{0},R_{1},H,W,0),
\\
O &=&A\cup B\cup \left\{ c,d,\#\right\} , \\
H &=&\left\{ l,l^{\prime }\mid l\in B\setminus \left\{  l_{h} \right\} \right\} 
\cup \left\{ l_{\langle x\rangle }\mid x\in \left\{ 1,2,1^{\prime },2^{\prime },d,\#\right\} \right\} ,
\end{eqnarray*}%
and the rules for $R_{1}$ and the sets of labels in $W$ are defined as
follows:

\medskip

\textbf{A.} Let $l_{i}:\left( \mathtt{ADD}\left( r\right)
,l_{j},l_{k}\right) $ be an ADD instruction in $I$. If $r>2$, then the
(labeled) rules 
\begin{equation*}
l_{i}:l_{i}\rightarrow l_{j}\left( a_{r},out\right) ,\ \ l_{i}^{\prime }:
l_{i}\rightarrow l_{k}\left( a_{r},out\right) ,
\end{equation*}%
are used, and for $r\in \left\{ 1,2\right\} $, we take the rules 
\begin{equation*}
l_{i}:l_{i}\rightarrow l_{j}a_{r},\ \ l_{i}^{\prime }:l_{i}\rightarrow
l_{k}a_{r}.
\end{equation*}%
In both cases, we define $\left\{ l_{i},l_{i}^{\prime }\right\} $ to be the
corresponding set of labels in $W$. The contents of each register $r$, $r\in
\left\{ 1,2\right\} $, is represented by the number of objects $a_{r}$
present in the skin membrane; any object $a_{r}$ with $r\geq 3$ is 
immediately sent out into the environment. \smallskip 

\textbf{B.} The simulation of a SUB instruction $l_{i}:\left( \mathtt{SUB}%
\left( r\right) ,l_{j},l_{k}\right) $, for $r\in \left\{ 1,2\right\} $, is
carried out by the following rules and the corresponding sets of labels in $W
$: 

For the case that the register $r$, $r\in \left\{ 1,2\right\} $, is not
empty we take the (labeled) rules 
\begin{equation*}
l_{i}:l_{i}\rightarrow l_{j},\ \ l_{\langle r\rangle }:ca_{r}\rightarrow c,\ \
l_{\langle d\rangle }:cd\rightarrow c\#,
\end{equation*}%
(if no symbol $a_{r}$ is present, i.e., if the register $r$ is empty, then
the trap symbol $\#$ is introduced by the rule $l_{\langle d\rangle }:
cd\rightarrow c\#$).

For the case that the register $r$ is empty, we take the (labeled) rules 
\begin{equation*}
l_{i}^{\prime }:l_{i}\rightarrow l_{k},\ \ l_{\langle r^{\prime }\rangle }:
ca_{r}\rightarrow c\#
\end{equation*}%
(if at least one symbol $a_{r}$ is present, i.e., if the register $r$ is not
empty, then the trap symbol $\#$ is introduced by the rule $l_{\langle 
r^{\prime }\rangle }:ca_{r}\rightarrow c\#$). 

The corresponding sets of labels to be taken into $W$ are 
$\left\{ l_{i},l_{\langle r\rangle },l_{\langle d\rangle }\right\} $ and 
$\left\{  l_{i}^{\prime },l_{\langle r^{\prime }\rangle }\right\} $, respectively. 
In both cases, the simulation of the SUB instruction works correctly if 
we have made the right choice.

\medskip

\textbf{C.} We also add the labeled rule 
$l_{\langle \# \rangle }:\#\rightarrow \#$ to $R_{1}$ and the set 
$\left\{ l_{\langle \# \rangle }\right\} $ to $W$, hence, the computation 
cannot halt once the trap symbol $\#$ has been generated.

In sum, we observe that each computation step in $M$ is simulated by 
exactly one computation step in $\Pi $; moreover, such a simulating 
computation in $\Pi $ halts if and only if the corresponding computation 
in $M$ halts (as soon as the label  $l_{h}$ appears, no rule can be 
applied anymore in $\Pi $, as we have not defined any rule for the HALT 
instruction of $M$). If at some moment we make the wrong choice when
trying to simulate a SUB instruction and have to generate the trap 
symbol $\# $, the computation will never halt. Hence, we have shown 
$Ps\left( M\right) =Ps\left( \Pi \right) $, which completes the proof.
$\hfill {}$
\end{proof}

\section{Computational Completeness of P Systems with Target Selection}

\begin{theorem}
\label{TheoremTS} $YOP_{7}\left( cat_{1},ts\right) =YRE$, $Y\in \left\{
N,Ps\right\} $.
\end{theorem}

\begin{proof}
We only prove the inclusion $PsRE\subseteq PsOP_{7}\left( cat_{1},ts\right) $%
. Let us consider a register machine $M=\left( n+2,B,l_{0},l_{h},I\right) $
with only the first and the second register ever being decremented, and let $%
A=\left\{ a_{1},\dots ,a_{n+2}\right\} $ be the set of objects for
representing the contents of the registers $1$ to $n+2$ of $M$. The set of
labels $B\setminus \left\{ l_{h}\right\} $ is divided into three disjoint
subsets:%
\begin{eqnarray*}
B_{+} &=&\left\{ l\mid l_{i}:\left( \mathtt{ADD}\left( r\right)
,l_{j},l_{k}\right) \in I\right\} , \\
B_{-r} &=&\left\{ l\mid l_{i}:\left( \mathtt{SUB}\left( r\right)
,l_{j},l_{k}\right) \in I\right\} ,\ r\in \left\{ 1,2\right\} ;
\end{eqnarray*}%
moreover, we define 
\begin{eqnarray*}
B_{-}&=&B_{-1}\cup B_{-2},\\ 
B_{-}^{\prime }&=&\left\{ l^{\prime }\mid l\in B_{-}\right\} ,\\ 
B_{-}^{\prime \prime }&=&\left\{ l^{\prime \prime }\mid l\in B_{-}\right\} , 
\rm{\ and}\\ 
B^{\prime }&=&B_{+}\cup B_{-}\cup B_{-}^{\prime }\cup B_{-}^{\prime \prime }.
\end{eqnarray*}

We construct the following P system: 
\begin{eqnarray*}
\Pi  &=&(O,\left\{ c\right\} ,[\ [\ \ {{\ ]_{{}}}_{{}}}_{2}\dots \lbrack \ \ 
{{\ ]_{{}}}_{{}}}_{7}{{\ ]_{{}}}_{{}}}_{1},w_{1},\dots ,w_{7},R_{1},\dots
,R_{7},0), \\
O &=&A\cup B^{\prime }\cup \left\{ a_{1}^{\prime },a_{2}^{\prime
},c,d,\#\right\} ,
\end{eqnarray*}%
with $w_{1}=l_{0}$, $w_{2}=c$, and $w_{i}=\lambda $ for $3\leq i\leq 7$. In
order to make argumentation easier, in the following we refer to the
membrane labels $1$ to $7$ according to the following table:%
\begin{equation*}
\begin{tabular}[b]{|l|l|l|l|l|l|l|}
\hline
$1$ & $2$ & $3$ & $4$ & $5$ & $6$ & $7$ \\ \hline
skin & $-$ & $0_{1}$ & $0_{2}$ & $-_{1}$ & $-_{2}$ & $+$ \\ \hline
\end{tabular}%
\end{equation*}

The sets of rules now are constructed as follows: \medskip

\textbf{A.} The simulation of any instruction from $I$ starts in the skin
membrane with moving all objects except the output symbols $a_{r}$, 
$3\leq r\leq n+2$, into an inner membrane; according to the definition, 
taking the target $in$ means choosing one of the inner membranes in a 
non-deterministic way, but the same membrane for all objects to be 
moved in. The output symbols $a_{r}$, $3\leq r\leq n+2$, are sent out into 
the environment by $a_{r}\rightarrow \left( a_{r},out\right) $, thus yielding 
the result of a halting computation as the number of symbols $a_{r}$ 
sent out into the environment during this computation. In case some 
copies of the output symbols $a_{r}$, $3\leq r\leq n+2$, are present in the 
skin membrane, at any time we may either select the target $out$ to send 
all these objects out into the environment or else select the target $in$
in order to start the simulation of the next instruction. Choosing the target
$out$ or $in$ always is done in a non-deterministic way.
Hence, in sum we get%
\begin{equation*}
\begin{array}[b]{lll}
R_{1} & = & \left\{ x\rightarrow \left( x,in\right) \mid x\in B_{+}\cup
B_{-}\cup \left\{ a_{1},a_{2},a_{1}^{\prime },a_{2}^{\prime },\#\right\}
\right\} \cup \left\{ x\rightarrow \left( xd,in\right) \mid x\in
B_{-}^{\prime }\right\}  \\ 
& \cup  & \left\{ a_{r}\rightarrow \left( a_{r},out\right) \mid 3\leq r\leq
n+2\right\} .%
\end{array}%
\end{equation*}

\textbf{B.} For the simulation of an ADD instruction $l_{i}:\left( \mathtt{%
ADD}\left( r\right) ,l_{j},l_{k}\right) \in I$ all non-terminal symbols (all
symbols except $a_{r}$, $r\geq 3$) are expected to have been sent to
membrane $+$:%
\begin{equation*}
\begin{array}[b]{lll}
R_{+} & = & \left\{ l_{i}\rightarrow \left( l_{j}a_{r},out\right)
,l_{i}\rightarrow \left( l_{k}a_{r},out\right) \mid l_{i}:\left( \mathtt{ADD}%
\left( r\right) ,l_{j},l_{k}\right) \in I\right\}  \\ 
& \cup  & \left\{ l\rightarrow \left( \#,out\right) \mid l\in B^{\prime
}\setminus B_{+}\right\}  \\ 
& \cup  & \left\{ x\rightarrow \left( x,out\right) \mid x\in \left\{
a_{1},a_{2},\#\right\} \right\} .%
\end{array}%
\end{equation*}%
If the symbols arrive in membrane $+$ with a label $l\in B^{\prime
}\setminus B_{+}$, then the trap symbol $\#$ is generated and the
computation will never halt.\medskip 

\textbf{C.} The simulation of a SUB instruction $l_{i}:\left( \mathtt{SUB}%
\left( r\right) ,l_{j},l_{k}\right) $ is carried out in two steps for the
zero test, i.e., when the register $r$ is empty, using (the rules in)
membrane $0_{r}$ and in five steps for decrementing the number of symbols $%
a_{r}$, first using membrane $-_{r}$ to mark the corresponding symbols $a_{r}
$ into $a_{r}^{\prime }$ and then using the catalyst $c$ in membrane $-$ to
erase one of these primed objects; the marking procedure is necessary to
guarantee that the catalyst erases the right object. For $r\in \left\{
1,2\right\} $, we define the following sets of rules:%
\begin{equation*}
\begin{array}[b]{lll}
R_{0_{r}} & = & \left\{ l_{i}\rightarrow \left( l_{k},out\right)
,a_{r}\rightarrow \left( \#,out\right) \mid l_{i}:\left( \mathtt{SUB}\left(
r\right) ,l_{j},l_{k}\right) \in I\right\}  \\ 
& \cup  & \left\{ l\rightarrow \left( \#,out\right) \mid l\in B^{\prime
}\setminus B_{-r}\right\}  \\ 
& \cup  & \left\{ x\rightarrow \left( x,out\right) \mid x\in \left\{
a_{3-r},\#\right\} \right\} .%
\end{array}%
\end{equation*}

If the number of objects $a_{r}$ is not zero, i.e., if the register $r$ is
not empty, the introduction of the trap symbol $\#$ causes the computation
to never halt. On the other hand, if we want to decrement the register, we
have to guarantee that exactly one symbol $a_{r}$ is erased:%
\begin{equation*}
\begin{array}[b]{lll}
R_{-_{r}} & = & \left\{ l_{i}\rightarrow \left( l_{i}^{\prime },out\right)
\mid l_{i}\in B_{-r}\right\} \cup \left\{ a_{r}\rightarrow \left(
a_{r}^{\prime },out\right) \right\}  \\ 
& \cup  & \left\{ l\rightarrow \left( \#,out\right) \mid l\in B^{\prime
}\setminus B_{-r}\right\}  \\ 
& \cup  & \left\{ x\rightarrow \left( x,out\right) \mid x\in \left\{
a_{3-r},\#\right\} \right\} .%
\end{array}%
\end{equation*}%
The whole multiset of objects via the skin membrane now has to enter
membrane $-$; here the dummy symbol $d$ guarantees that the catalyst cannot
do nothing if no primed symbol $a_{r}^{\prime }$ has arrived; again the
generation of $\#$ causes the computation to not halt anymore:%
\begin{equation*}
\begin{array}[b]{lll}
R_{-} & = & \left\{ l_{i}^{\prime }\rightarrow l_{j}^{\prime \prime
},ca_{r}^{\prime }\rightarrow c,l_{i}^{\prime \prime }\rightarrow
\#,l_{i}^{\prime \prime }\rightarrow \left( l_{j},out\right) \mid
l_{i}:\left( \mathtt{SUB}\left( r\right) ,l_{j},l_{k}\right) \in I\right\} 
\\ 
& \cup  & \left\{ cd\rightarrow c\#,d\rightarrow \left( \lambda ,out\right)
\right\} \cup \left\{ a_{r}^{\prime }\rightarrow \left( a_{r},out\right)
\mid r\in \left\{ 1,2\right\} \right\} , \\ 
& \cup  & \left\{ l\rightarrow \left( \#,out\right) \mid l\in B^{\prime
}\setminus B_{-}^{\prime \prime }\right\}  \\ 
& \cup  & \left\{ x\rightarrow \left( x,out\right) \mid x\in \left\{
a_{3-r},\#\right\} \right\} .%
\end{array}%
\end{equation*}%
The end of the simulation of the SUB instruction $l_{i}:\left( \mathtt{SUB}%
\left( r\right) ,l_{j},l_{k}\right) $ in membrane $-$ takes two steps:
first,we apply $l_{i}^{\prime }\rightarrow l_{j}^{\prime \prime }$ and $%
ca_{r}^{\prime }\rightarrow c$, thus erasing exactly one symbol $%
a_{r}^{\prime }$, which corresponds to decrement register $r$; in the second
step, we send out the label $l_{j}$ by using $l_{i}^{\prime \prime
}\rightarrow \left( l_{j},out\right) $ together with the remaining symbols $%
a_{r}$ by using $a_{r}^{\prime }\rightarrow \left( a_{r},out\right) $ and
all symbols $a_{3-r}$ by using $a_{3-r}\rightarrow \left( a_{3-r},out\right) 
$. The additional symbol $d$ generated in the first step in the skin
membrane is eliminated by applying the rule $d\rightarrow \left( \lambda
,out\right) $. These two steps cannot be interchanged, as with using the
target $out$ first we would have to use the rule $l_{i}^{\prime }\rightarrow
\left( \#,out\right) $, thus introducing the trap symbol $\#$.

If in any of the membranes $R_{0_{r}}$, $R_{-_{r}}$, $r\in \left\{ 1,2\right\}$, 
and $R_{-} $ the symbols arrive with the wrong label $l\in B^{\prime }$, 
then the trap symbol $\#$ is generated and the computation will never halt.

We finally observe that a computation in $\Pi $ halts if and only if the
final label $l_{h}$ appears (and then stays in the skin membrane) and no
trap symbol $\#$ is present, hence, we conclude $Ps\left( M\right) =Ps\left(
\Pi \right) $.
$\hfill {}$
\end{proof}
\medskip

To eventually reduce the number of inner membranes remains as a challenging
task for future research.

\section{Computational Completeness of Time-Varying P Systems}

\begin{theorem}
\label{TheoremTV} $YOP_{1}\left( cat_{1},\alpha TV_{6}\right) =YRE$, $\alpha
\in \left\{ \lambda ,w\right\} $, $Y\in \left\{ N,Ps\right\} $.
\end{theorem}

\begin{proof}
We only prove the inclusion $PsRE\subseteq PsOP_{1}\left(
cat_{1},TV_{6}\right) $. Let us consider a register machine $M=\left(
n+2,B,l_{0},l_{h},I\right) $ with only the first and the second register
ever being decremented. Again, we define $A=\left\{ a_{1},\dots
,a_{n+2}\right\} $ and divide the set of labels $B\setminus \left\{
l_{h}\right\} $ into three disjoint subsets:%
\begin{eqnarray*}
B_{+} &=&\left\{ l_{i}\mid l_{i}:\left( \mathtt{ADD}\left( r\right)
,l_{j},l_{k}\right) \in I\right\} , \\
B_{-r} &=&\left\{ l_{i}\mid l_{i}:\left( \mathtt{SUB}\left( r\right)
,l_{j},l_{k}\right) \in I\right\} ,\text{ }r\in \left\{ 1,2\right\} ;
\end{eqnarray*}%
moreover, we define $B_{-}=B_{-1}\cup B_{-2}$ as well as 
\begin{equation*}
B^{\prime }=\left\{ l,\tilde{l},\hat{l}\mid l\in B\setminus \left\{
l_{h}\right\} \right\} \cup \left\{ l^{-},l^{0},\bar{l}^{-},\bar{l}^{0},\mid
l\in B_{-}\right\} .
\end{equation*}%
The main challenge in the construction for the time-varying P system $\Pi $
is that the catalyst has to fulfill its task to erase an object $a_{r}$, $%
r\in \left\{ 1,2\right\} $, for both objects in the same membrane where all
other computations are carried out, too; hence, at a specific moment in the
cycle of period six, parts of simulations of different instructions have to
be coordinated in parallel. The basic components of the time-varying P
system $\Pi $ are defined as follows (we here do not distinguish between a
rule and its label): 
\begin{eqnarray*}
\Pi  &=&(O,\left\{ c\right\} ,[\ {{\ ]_{{}}}_{{}}}_{1},cl_{0},R_{1}\cup
\dots \cup R_{6},R_{1}\cup \dots \cup R_{6},\left( R_{1}\dots R_{6}\right)
^{\ast },0), \\
O &=&A\cup \left\{ a_{1}^{\prime },a_{2}^{\prime }\right\} \cup B^{\prime
}\cup \left\{ c,h,l_{h},\#\right\} .
\end{eqnarray*}

We now list the rules in the sets of rules $R_{i}$ to be applied in
computation steps $6n+i$, $n\geq 0$, $1\leq i\leq 6$:

\medskip

$\mathbf{R}_{1}$: in this first step of the cycle, especially all the ADD 
instructions are simulated, i.e., for each 
$l_{i}:\left( \mathtt{ADD}\left( r\right) ,l_{j},l_{k}\right) \in I$ we take

$cl_{i}\rightarrow ca_{r}\tilde{l}_{j}$, $cl_{i}\rightarrow ca_{r}\tilde{l}%
_{k}$ for $r\in \left\{ 1,2\right\} $ as well as $cl_{i}\rightarrow c(a_{r},out)\tilde{l}_{j}$, $cl_{i}\rightarrow c(a_{r},out)\tilde{l}%
_{k}$ for $3\leq r\leq n+2 $ (in order to obtain the output in the environment, for $r\geq 3$ we have to take $(a_{r},out)$ instead of  $a_{r}$);
only in the sixth step of the cycle, from $\tilde{l}_{j}$ and $\tilde{%
l}_{k}$ the corresponding unmarked labels $l_{j}$ and $l_{k}$ will be
generated;

$cl\rightarrow cl^{-}$, $cl\rightarrow cl^{0}$ initiate the simulation of a
SUB instruction for register $1$ labeled by $l\in B_{-1}$, i.e., we make a
non-deterministic guess whether register $r$ is empty (with introducing $%
l^{0}$) or not (with introducing $l^{-}$);

$cl\rightarrow c\hat{l}$ marks a label $l\in B_{-2}$ (the simulation of such
a SUB instruction for register $2$ will start in step $4$ of the cycle);

$\#\rightarrow \#$ keeps the trap symbol $\#$ alive guaranteeing an infinite
loop once $\#$ has been generated;

$h\rightarrow \lambda $ eliminates the auxiliary object $h$ which 
eventually has been generated two steps before ($h$ is needed for
simulating the decrement case of SUB instructions).

\medskip

$\mathbf{R}_{2}$: in the second and the third step, the SUB\ instructions on
register $1$ are simulated, i.e., for all $l\in B_{-1}$ we start with

$ca_{1}\rightarrow ca_{1}^{\prime }$ (if present, exactly one copy of $a_{1}$
can be primed, but only if a label $l^{-}$ for some $l$ from $B_{-1}$ is
present) and

$l^{-}\rightarrow \bar{l}^{-}h$, $l^{0}\rightarrow \bar{l}^{0}$ for all $%
l\in B_{-1}$;

all other labels $\tilde{l}$ for $l\in B$ block the catalyst $c$ from
erasing a copy of $a_{1}$ by forcing the application of the corresponding
rules $c\tilde{l}\rightarrow c\tilde{l}$ for $c$ in order to avoid the
introduction of the trap symbol $\#$ by the enforced application of a rule $%
\tilde{l}\rightarrow \#$, i.e., we take

$c\tilde{l}\rightarrow c\tilde{l}$, $\tilde{l}\rightarrow \#$ for all $l\in B
$, and

$c\hat{l}\rightarrow c\hat{l}$, $\hat{l}\rightarrow \#$ for all $l\in B_{-2}$%
;

$\#\rightarrow \#$ keeps the computation alive once the trap symbol has been
introduced.

\medskip

$\mathbf{R}_{3}$: for all $l_{i}:\left( \mathtt{SUB}\left( 1\right)
,l_{j},l_{k}\right) \in I$ we take

$c\bar{l}_{i}^{0}\rightarrow c\tilde{l}_{k}$, $a_{1}^{\prime }\rightarrow \#$%
, $\bar{l}_{i}^{0}\rightarrow \#$ (zero test; if a primed copy of $a_{1}$ is
present, then the trap symbol $\#$ is generated);

$\bar{l}_{i}^{-}\rightarrow \tilde{l}_{j}$, $ca_{1}^{\prime }\rightarrow c$, 
$ch\rightarrow c\#$ (decrement; the auxiliary symbol $h$ is needed to keep
the catalyst $c$ busy with generating the trap symbol $\#$ if we have taken
the wrong guess when assuming the register $1$ to be non-empty);

$c\tilde{l}\rightarrow c\tilde{l}$, $\tilde{l}\rightarrow \#$ for all $l\in B
$ (with these labels, we just pass through this step);

$c\hat{l}\rightarrow c\hat{l}$, $\hat{l}\rightarrow \#$ for all $l\in B_{-2}$
(these labels pass through this step to become active in the next step);

$\#\rightarrow \#$.

\medskip

$\mathbf{R}_{4}$: in the fourth step, the simulation of SUB instructions on
register $2$ is initiated by using

$c\hat{l}\rightarrow cl^{-}$, $c\hat{l}\rightarrow cl^{0}$ for all $l\in
B_{-2}$, i.e., we make a non-deterministic guess whether register $r$ is
empty (with introducing $l^{0}$) or not (with introducing $l^{-}$);

$c\tilde{l}\rightarrow c\tilde{l}$, $\tilde{l}\rightarrow \#$ for all $l\in B
$ (with all other labels, we only pass through this step);

$\#\rightarrow \#$,

$h\rightarrow \lambda $ (if $h$ has been introduced by $l^{-}\rightarrow 
\bar{l}^{-}h$ in the second step for some $l\in B_{-1}$, we now erase it).

\medskip

$\mathbf{R}_{5}$: in the fifth and the sixth step, the SUB\ instructions on
register $2$ are simulated, i.e., for all $l\in B_{-2}$ we start with

$ca_{2}\rightarrow ca_{2}^{\prime }$ (if present, exactly one copy of $a_{2}$
can be primed) and

$l^{-}\rightarrow \bar{l}^{-}h$, $l^{0}\rightarrow \bar{l}^{0}$ for all $%
l\in B_{-2}$;

$c_{1}\tilde{l}\rightarrow c_{1}\tilde{l}$, $\tilde{l}\rightarrow \#$ for
all $l\in B$;

$\#\rightarrow \#$.

\medskip

$\mathbf{R}_{6}$: the simulation of SUB instructions $l_{i}:\left( \mathtt{%
SUB}\left( 2\right) ,l_{j},l_{k}\right) \in I$ on register $2$ is finished by

$c\bar{l}_{i}^{0}\rightarrow cl_{k}$, $a_{2}^{\prime }\rightarrow \#$, $\bar{%
l}_{i}^{0}\rightarrow \#$ (zero test; if a primed copy of $a_{2}$ is
present, then the trap symbol $\#$ is generated);

$\bar{l}_{i}^{-}\rightarrow l_{j}$, $ca_{2}^{\prime }\rightarrow c$, $%
ch\rightarrow c\#$ (decrement; the auxiliary symbol $h$ is needed to keep
the catalyst $c$ busy with generating the trap symbol $\#$ if we have taken
the wrong guess when assuming the register $2$ to be non-empty; if it is not
used, it can be erased in the next step by using $h\rightarrow \lambda $ in $%
R_{1}$);

$c\tilde{l}\rightarrow cl$, $\tilde{l}\rightarrow \#$ for all $l\in B$;

$\#\rightarrow \#$ .

\medskip

Without loss of generality, we may assume that the final label $l_{h}$ in $M$
is only reached by using a zero test on register $2$; then, at the beginning
of a new cycle, after a correct simulation of a computation from $M$ in the
time-varying P system $\Pi $ no rule will be applicable in $R_{1}$ (another
possibility would be to take $c\bar{l}_{i}^{0}\rightarrow c$ instead of $c%
\bar{l}_{i}^{0}\rightarrow cl_{h}$ in $R_{6}$).

At the end of the cycle, in case all guesses have been correct, the
requested instruction of $M$ has been simulated and the label of the next
instruction to be simulated is present in the skin membrane. Only in the
case that $M$ has reached the final label $l_{h}$, the computation in $\Pi $
halts, too, but only if during the simulation of the computation of $M$ in $%
\Pi $ no trap symbol $\#$ has been generated; hence, we conclude $Ps\left(
M\right) =Ps\left( \Pi \right) $.
$\hfill {}$
\end{proof}

\section{Computational Completeness of P Systems with Membrane Creation}

\begin{theorem}
\label{mcre} $YOP_{1,2}\left( cat_{1},mcre\right) =YRE$, $Y\in \left\{
N,Ps\right\} $.
\end{theorem}

\begin{proof}
We only prove the inclusion $PsRE\subseteq PsOP_{1,2}\left(
cat_{1},mcre\right) $. Let us consider a register machine $M=\left(
n+2,B,l_{0},l_{h},I\right) $ with only the first and the second register
ever being decremented. Again we define $A=\left\{ a_{1},\dots
,a_{n+2}\right\} $ as the set of objects for representing the contents of
the registers $1$ to $n+2$ of $M$. We construct the following P system: 
\begin{eqnarray*}
\Pi  &=&\left( O,\left\{ c\right\} ,[\ \ {{\ ]_{{}}}_{{}}}%
_{1},cdl_{0},R_{1},R_{2},R_{3},0\right) , \\
O &=&A\cup \left\{ l,l^{\prime },l^{\prime \prime }\mid l\in B\right\} \cup
\left\{ c,d,d^{\prime },d^{\prime \prime }\right\} ,
\end{eqnarray*}%
and the sets of rules are constructed as follows. \medskip 

\textbf{A.} For each ADD instruction $l_{i}:\left( \mathtt{ADD}\left(
r\right) ,l_{j},l_{k}\right) $ in $I$, the rules 
\begin{eqnarray*}
&\mbox{step 1: }&l_{i}\rightarrow l_{i}^{\prime },\ d\rightarrow d^{\prime },
\\
&\mbox{step 2: }&l_{i}^{\prime }\rightarrow a_{r}l_{j},\ l_{i}^{\prime
}\rightarrow a_{r}l_{k},\ d^{\prime }\rightarrow d,
\end{eqnarray*}%
are taken into $R_{1}$ and obviously simulate an ADD instruction in two
steps. We also add the rules $a_{r}\rightarrow \left( a_{r},out\right) $ for 
$3\leq r\leq n+2$ to $R_{1}$; thus, in any moment, every copy of $a_{r}$, 
$3\leq r\leq n+2$, present in the skin membrane is sent out to the environment. \smallskip 

\textbf{B.} For each SUB instruction $l_{i}:\left( \mathtt{SUB}\left(
r\right) ,l_{j},l_{k}\right) $ in $I$, the following rules in 
$R_{1}$ and $R_{r+1}$, $r\in \left\{ 1,2\right\} $, are used:%
\begin{equation*}
\begin{tabular}{c|c|c}
Step & $R_{1}$ & $R_{r+1}$ \\ \hline
1 & $cl_{i}\rightarrow c[\ l_{i}{{\ ]_{{}}}_{{}}}_{r+1},\ d\rightarrow
d^{\prime }$ & -- \\ 
2 & $\quad ca_{r}\rightarrow c\left( a_{r},in_{r+1}\right) ,\ d^{\prime
}\rightarrow \left( d^{\prime },in_{r+1}\right) \quad $ & $\quad
l_{i}\rightarrow l_{i}^{\prime }$ \\ 
3 & -- & $\quad a_{r}\rightarrow \lambda \delta ,\ l_{i}^{\prime
}\rightarrow l_{i}^{\prime \prime },\ d^{\prime }\rightarrow d^{\prime
\prime }$ \\ 
4 & $cl_{i}^{\prime \prime }\rightarrow cl_{j},\ d^{\prime \prime
}\rightarrow d$ & $l_{i}^{\prime \prime }\rightarrow l_{k},\ d^{\prime
\prime }\rightarrow d\delta $%
\end{tabular}%
\end{equation*}%
A SUB instruction $l_{i}:\left( \mathtt{SUB}\left( r\right)
,l_{j},l_{k}\right) $ (with $r\in \left\{ 1,2\right\} $) is simulated
according to the four steps suggested in the table given above:

In the first step, we create a membrane with the label $r+1$, where $l_{i}$
is sent to, and simultaneously $d$ becomes $d^{\prime }$. In the next step,
if any $a_{r}$ exists, i.e., if register $r$ is not empty, then one copy of $%
a_{r}$ should enter the membrane $r+1$ just having been created in the
preceding step. Note that the selection of the right membrane (the use of $%
in_{r+1} $ instead of $in$) is important: $a_{r}$ has to go to the membrane
created in the previous step, when $r+1$ has been specified by the label $%
l_{i}$. At the same time, $d^{\prime }$ enters the membrane $r+1$, and $%
l_{i} $ becomes $l_{i}^{\prime }$ in this membrane. If the register $r$ is
empty, then the catalyst is doing nothing in this second step.

In the third step, in membrane $r+1$, $l_{i}^{\prime }$ becomes $%
l_{i}^{\prime \prime }$ and $d^{\prime }$ becomes $d^{\prime \prime }$. If $%
a_{r}$ is not present in membrane $r+1$, nothing else happens there in this
step; if $a_{r}$ is present, it dissolves the membrane and disappears.
Observe that in both cases $ca_{r}\rightarrow c\left( a_{r},in_{r+1}\right) $
will not be applicable (anymore) in $R_{1}$. Thus, we either have $%
cl_{i}^{\prime \prime }d^{\prime \prime }$ in the skin membrane (when the
register has been non-empty), or we have only $c$ in the skin membrane and $%
l_{i}^{\prime \prime }d^{\prime \prime }$ in the inner membrane $r+1$. In
the first case, in the fourth step we use the rules $cl_{i}^{\prime \prime
}\rightarrow cl_{j}\ $and $d^{\prime \prime }\rightarrow d$ from $R_{1}$,
which is the correct continuation of the simulation of the SUB instruction;
in the latter case, we use $l_{i}^{\prime \prime }\rightarrow l_{k}$ and $%
d^{\prime \prime }\rightarrow d\delta $ in $R_{r+1}$. The inner membrane is
dissolved, and in the skin membrane we get the objects $cl_{k}d$. In both
cases, the simulation of the SUB instruction is correct and we return to a
configuration as that we started with, hence, the simulation of another
instruction can start. \smallskip

There is one interference between the rules of $\Pi $ simulating the ADD and
the SUB instructions of $M$. If in the second step of simulating a SUB
instruction, instead of $d^{\prime }\rightarrow \left( d^{\prime
},in_{r+1}\right) $ we use $d^{\prime }\rightarrow d$, then the case when
register $r$ is non-empty continues correctly, as the simulation lasts four
steps, and in the end $d$ is present in the skin membrane (the dissolution
of membrane $r+1$ is done by $a_{r}$). If the register $r$ has been empty, $%
l_{i}^{\prime \prime }$ will become $l_{k}$ in membrane $r+1$ and it will
remain there until $d^{\prime }$ enters the membrane, changes to $d^{\prime
\prime }$, and then dissolves it (as long as $d,d^{\prime }$ switch to each
other in the skin membrane, the computation cannot halt). Thus, also in this
case we have to return to the correct submultiset $cl_{k}d$ in the skin 
membrane.

Consequently, exactly the halting computations of $M$ are simulated by the
halting computations in $\Pi $; hence, $Ps\left( M\right) =Ps\left( \Pi
\right) $. The observation that the maximal number of membranes in any
computation of $\Pi $ is two completes the proof.
$\hfill {}$
\end{proof}

\medskip

It remains as an open problem whether it is possible to use the target $in$
only instead of the $in_{j}$.

\section{Computational Completeness of P Systems with Mobile Catalysts}

If the membrane creation rules are of the form $ca\rightarrow \lbrack \ cb{{%
\ ]_{{}}}_{{}}}_{i}$, then this implicitly means that the catalyst is moving
from one region to another one. However, for mobile catalysts, the
computational completeness of such systems with only one catalyst has
already been proved in \cite{KrishnaAndrei}, using three membranes and
targets of the forms $here$, $out$, and $in_{j}$. In this paper, we improve
this result from the last point of view, making only use of the targets $here
$, $out$, and $in$. In fact, if in the proof of Theorem~\ref{TheoremTS} we
let the catalyst $c$ move with all the other objects, then we immediately
obtain a proof for $NOP_{7}\left( mcat_{1}\right) =NRE$ where even only the
targets $out$ and $in$ are used (but instead of three we need seven
membranes).

\begin{theorem}
\label{mobile} $YOP_{3}\left( mcat_{1}\right) =YRE$, $Y\in \left\{
N,Ps\right\} $.
\end{theorem}

\begin{proof}
We only prove the inclusion $PsRE\subseteq PsOP_{3}\left( mcat_{1}\right) $.
Let us consider a register machine $M=\left( n+2,B,l_{0},l_{h},I\right) $
with only the first and the second register ever being decremented. Again we
define $A=\left\{ a_{1},\dots ,a_{n+2}\right\} $ as the set of objects for
representing the contents of the registers $1$ to $n+2$ of $M$. We construct
the following P system: 

\begin{eqnarray*}
\Pi  &=&(O,\left\{ c\right\} ,[\ [\ \ {{\ ]_{{}}}_{{}}}_{2}[\ \ {{\ ]_{{}}}%
_{{}}}_{3}{{\ ]_{{}}}_{{}}}_{1},cl_{0},R_{1},R_{2},R_{3},0), \\
O &=&A\cup \left\{ l,l^{\prime },l^{\prime \prime },l^{\prime \prime \prime
}\mid l\in B\right\} \cup \left\{ c,\#\right\} , \\
R_{1}&=& \left\{  l_{i}\rightarrow l_{j}\left( a_{r},in\right) ,\ 
               l_{i}\rightarrow l_{k}\left( a_{r},in\right)   
               \mid l_{i}:\left( \mathtt{ADD}\left( r\right) ,l_{j},l_{k}\right) \in I,
               r\in \left\{ 1,2\right\} \right\} \\
&\cup &  \left\{  l_{i}\rightarrow l_{j}\left( a_{r},out\right) ,\ 
               l_{i}\rightarrow l_{k}\left( a_{r},out\right)    
               \mid l_{i}:\left( \mathtt{ADD}\left( r\right) ,l_{j},l_{k}\right) \in I,
               3\leq r\leq n+2\right\} \\
&\cup & \left\{  cl_{i}\rightarrow \left( c,in\right) \left( l_{i},in\right)  
               \mid l_{i}:\left( \mathtt{SUB}\left( r\right) ,l_{j},l_{k}\right) \in I,
               r\in \left\{ 1,2\right\} \right\} \\
&\cup &  \left\{  cl_{i}^{\prime \prime \prime }\rightarrow cl_{j},\ 
               l_{i}^{\prime \prime \prime }\rightarrow \# ,\ 
               \#\rightarrow \#
               \mid l_{i}:\left( \mathtt{SUB}\left( r\right) ,l_{j},l_{k}\right) \in I,
               r\in \left\{ 1,2\right\} \right\} ,\\
R_{2}&=& \left\{  a_{2}\rightarrow \# ,\  \#\rightarrow \# ,\ 
               ca_{1}\rightarrow \left( c,out\right) \right\} \\
&\cup &  \left\{ l_{i}\rightarrow \# 
                \mid l_{i}:\left( \mathtt{SUB}\left( r\right) ,l_{j},l_{k}\right) \in I,
               r\in \left\{ 1,2\right\} \right\} \\
&\cup &  \left\{ cl_{i}\rightarrow cl_{i}^{\prime } ,\ 
               l_{i}^{\prime }\rightarrow l_{i}^{\prime \prime } ,\ 
               cl_{i}^{\prime \prime }\rightarrow \left( c,out\right) \left( l_{k},
               out\right)  ,\ 
               l_{i}^{\prime \prime }\rightarrow \left( l_{i}^{\prime \prime \prime },
               out\right) 
               \mid l_{i}:\left( \mathtt{SUB}\left( 1\right) ,l_{j},l_{k}\right) \in I \right\} ,\\
R_{3}&=&  \left\{  a_{1}\rightarrow \# ,\  \#\rightarrow \# ,\ 
               ca_{2}\rightarrow \left( c,out\right) \right\} \\
&\cup &  \left\{ l_{i}\rightarrow \# 
                \mid l_{i}:\left( \mathtt{SUB}\left( r\right) ,l_{j},l_{k}\right) \in I,
               r\in \left\{ 1,2\right\} \right\} \\
&\cup &  \left\{ cl_{i}\rightarrow cl_{i}^{\prime } ,\ 
               l_{i}^{\prime }\rightarrow l_{i}^{\prime \prime } ,\ 
               cl_{i}^{\prime \prime }\rightarrow \left( c,out\right) \left( l_{k},
               out\right)  ,\ 
               l_{i}^{\prime \prime }\rightarrow \left( l_{i}^{\prime \prime \prime },
               out\right) 
               \mid l_{i}:\left( \mathtt{SUB}\left( 2\right) ,l_{j},l_{k}\right) \in I \right\} .
\end{eqnarray*}%
\noindent The rules in the sets of rules $R_{1}$, $R_{2}$, and $R_{3}$ are used 
as follows: 
\medskip 

\textbf{A.} Let $l_{i}:\left( \mathtt{ADD}\left( r\right)
,l_{j},l_{k}\right) $ be an ADD instruction in $I$. If $r\geq 3$, then the
rules $l_{i}\rightarrow l_{j}\left( a_{r},out\right) ,\ \ l_{i}\rightarrow
l_{k}\left( a_{r},out\right) $ are used in $R_{1}$; if $r\in \left\{
1,2\right\} $, in $R_{1}$ we take the rules $l_{i}\rightarrow
l_{j}\left( a_{r},in\right) $ and $l_{i}\rightarrow l_{k}\left(
a_{r},in\right) $ as well as the rules $a_{j}\rightarrow \#\ $and $%
\#\rightarrow \#$ in $R_{4-j}$, $j\in \left\{ 1,2\right\} $. The contents of
each register $r$, $r\in \left\{ 1,2\right\} $, is represented by the number
of objects $a_{r}$ present in membrane $r+1$; any object $a_{r}$, $r\geq 3$,
is immediately sent out into the environment. If $a_{j}$ is introduced in
membrane $4-j$, $j\in \left\{ 1,2\right\} $, then the trap object $\#$ is
produced and the computation never halts. \smallskip 

\textbf{B.} The simulation of a SUB instruction $l_{i}:\left( \mathtt{SUB}%
\left( r\right) ,l_{j},l_{k}\right) $ is carried out by the following rules
(the simulation again has four steps, as in the proof of Theorem \ref{mcre}):

For the first step, we take the rule $cl_{i}\rightarrow \left( c,in\right)
\left( l_{i},in\right) $ in $R_{1}$ and the rule $l_{i}\rightarrow \# $ in
both $R_{2}$ and $R_{3}$ (if $c$ and $l_{i}$ are not moved together into an
inner membrane, then the trap object $\#$ is produced and the computation
never halts).
In the second step, $R_{r+1}$ has to use the rule $cl_{i}\rightarrow
cl_{i}^{\prime }$. This checks whether $c$ and $l_{i}$ have been moved
together into the right membrane $r+1$; if this is not the case, then the
rule $cl_{i}\rightarrow cl_{i}^{\prime }$ is not available and the rule $%
l_{i}\rightarrow \#$ must be used, which causes the computation to never
halt.

Thus, after the second step, we know whether both $c$ and $l_{i}$ ($%
l_{i}^{\prime }$) are in the right membrane $r+1$. The rules $%
ca_{r}\rightarrow \left( c,out\right) $ and $l_{i}^{\prime }\rightarrow
l_{i}^{\prime \prime }$ in $R_{r+1}$ are used in order to perform the
third step of the simulation. If there is any copy of $a_{r}$ in membrane $%
r+1$ (i.e., if register $r$ is not empty), then the catalyst exits, while
also removing a copy of $a_{r}$. Simultaneously, $l_{i}^{\prime }$ becomes $%
l_{i}^{\prime \prime }$. Hence, if the register $r$ has been non-empty, we
now have $c$ in the skin membrane and $l_{i}^{\prime \prime }$ in membrane $%
r+1$; if register $r$ has been empty, we have both $c$ and $l_{i}^{\prime
\prime }$ in membrane $r+1$. We then use the rules $cl_{i}^{\prime \prime
}\rightarrow \left( c,out\right) \left( l_{k},out\right) $ and $l_{i}^{\prime
\prime }\rightarrow \left( l_{i}^{\prime \prime \prime },out\right) $ in $%
R_{r+1}$ as well as the rules $cl_{i}^{\prime \prime \prime }\rightarrow cl_{j}$
and $l_{i}^{\prime \prime \prime }\rightarrow \#$ in $R_{1}$. 
If $c$ is inside membrane $r+1$, we get $cl_{k}$ in the skin membrane,
which is the correct continuation for the case when the register is empty.
If $c$ is not in membrane $r+1$, then $l_{i}^{\prime \prime }$ exits alone
thereby becoming $l_{i}^{\prime \prime \prime }$, and, together with $c$,
which waits in the skin membrane, introduces $l_{j}$, which is a correct
continuation, too. If the rule $l_{i}^{\prime \prime }\rightarrow \left(
l_{i}^{\prime \prime \prime },out\right) $ is used although $c$ is inside
membrane $r+1$, then in the skin membrane we have to use the rule $%
l_{i}^{\prime \prime \prime }\rightarrow \#$ and the computation never halts
(as we have the rule $\# \rightarrow \# $ in $R_1$).

In all cases, the simulation of the SUB instruction works correctly, and we
return to a configuration with the catalyst and a label from $H$ in the skin
region. \medskip

In sum, we have the equality $Ps\left( M\right) =Ps\left( \Pi \right) $,
which completes the proof.
$\hfill {}$
\end{proof}

\section{Final Remarks}

Although we have exhibited several new computational completeness results
for P systems using only one catalyst together with some additional control
mechanism, the original problem of characterizing the sets of (vectors of)
non-negative integers generated by P systems with only one catalyst still
remains open. A similar challenging problem is to consider \textit{purely
catalytic} P systems with only two catalysts: with only one catalyst, we
obtain the regular sets; as shown in \cite{Freundetal2005}, three catalysts
are enough to obtain computational completeness. With two catalysts and some
additional control mechanism, computational completeness can be obtained,
too, see \cite{Freund2013}.

\bigskip

\noindent \textbf{Acknowledgements.} The work of Gheorghe P\u{a}un has been
supported by Proyecto de Excelencia con Investigador de Reconocida Val\'{\i}%
a, de la Junta de Andaluc\'{\i}a, grant P08 -- TIC 04200. The authors
gratefully acknowledge the suggestions of the referees.

\bibliographystyle{eptcs}
\bibliography{MCU2013FreundPaunFinal}

\end{document}